\documentclass[journal,10pt,onecolumn,draftclsnofoot]{IEEEtran}

\usepackage{amsmath}
\usepackage{amssymb}
\usepackage{amsfonts}
\usepackage{bm}

\usepackage{amsthm}
\usepackage{overpic}
\usepackage{lipsum}
\usepackage{enumitem}
  {
      \theoremstyle{plain}
      \newtheorem{assumption}{Assumption}
  }

\usepackage[left=0.75in,right=0.75in,top=0.75in,bottom=.75in]{geometry}
\usepackage[utf8]{inputenc}
\usepackage[T1]{fontenc}
\usepackage{url}
\usepackage{ifthen}
\usepackage{cite}
\usepackage{graphicx,epstopdf,amssymb,amsthm}% ,epstopdf,spconf,epsfig,fullpage}%, ,graphicx,psfrag,url,amsmath,amsthm,comment}
\usepackage{ulem,color}
\usepackage{epstopdf}
\usepackage[utf8]{inputenc}
\usepackage[english]{babel}
\usepackage{caption}

\DeclareCaptionLabelFormat{lc}{\MakeLowercase{#1}~#2}
\newtheorem{theorem}{Theorem}%[section]
\newtheorem{remark}{Remark}
\newtheorem{lemma}{Lemma}

\newtheorem{definition}{Definition}%[section]
\newtheorem{proposition}{Proposition}
\newtheorem{example}{Example}[]

\begin{document}
\graphicspath{{figs/}}

\title{The Effect of Population Flow on Epidemic Spread: Analysis and Control}


\author{Brooks Butler, Ciyuan Zhang, Ian Walter, Nishant Nair, Raphael Stern,
        and~Philip E. Par\'e*%~\IEEEmembership{Life~Fellow,~IEEE}% <-this % stops a space
\thanks{*Brooks Butler, Ciyuan Zhang, Ian Walter, Nishant Nair,
 and~Philip. E. Par\'e are with the School of Electrical and Computer Engineering at Purdue University. Emails: \{brooksbutler, zhan3375, walteri, nair65, philpare\}@purdue.edu. Raphael Stern is with the Department of Civil, Environmental, and Geo- Engineering at the University of Minnesota, Email: \{rstern@umn.edu\}.
%  }\thanks{*
 This work was funded in part by the C3.ai Digital Transformation Institute sponsored by C3.ai Inc. and the Microsoft Corporation and in part 
   by the National Science Foundation, grants NSF-CNS \#2028738 (P.E.P.), NSF-CNS \#2028946 (R.S.), and NSF-ECCS \#2032258 (P.E.P).}}
%\thanks{*This work was funded by Microsoft and C3.ai.}

\maketitle

\begin{abstract}
In this paper, we present a discrete-time networked SEIR model using population flow, 
its derivation, and assumptions under which this model is well defined. We identify properties of the system's equilibria, namely the healthy states. We show that the 
set of healthy states is
asymptotically stable, and that the value of the equilibria
becomes equal across all sub-populations as a result of the network flow model. Furthermore, we explore closed-loop feedback control of the system by limiting flow between sub-populations as a function of the current infected states. These results are illustrated via simulation based on flight traffic between major airports in the United States. 
We find that a flow restriction strategy combined with a vaccine roll-out significantly reduces the total number of infections over the course of an epidemic, given that the initial flow restriction response is not delayed.
\end{abstract}

\IEEEpeerreviewmaketitle

%%%%%%%%%%%%%%%%%%%%%%%%%%%%%%%%%%%%%%%%%%%%%%%%%%%%%%%%%%%%%%%%%%%%%%
\section{Introduction}
Global interconnectivity has proven to be a key factor in the propagation of infectious diseases~\cite{ruan2006effect, tatem2006global}. Most recently, we have seen evidence of such connectivity through the rapid spread of the COVID-19 pandemic, which propagated 
% \phil{spread [we have `spread' in this sentence twice]} 
from its origin in Wuhan, China to every major population center globally in a matter of weeks \cite{rodriguez2020going}. Given the implications of global population flow on disease spread, it becomes important to accurately model 
% such networked 
this
flow, as reliable modeling is an essential step to developing effective and efficient mitigation strategies.
% \phil{Epidemic model development is used to identify conditions to eradicate the pathogen and leverage the knowledge of these conditions to design mitigation strategies [this sentence seems pretty close to the last sentence...].} 
Various infection models have been proposed based on characteristics of individual pathogens and studied in the literature, including susceptible-infected-susceptible (SIS), susceptible-infected-removed (SIR), and susceptible-infected-removed-susceptible (SIRS)~\cite{rock2014dynamics, mei2017dynamics}. For this paper, we consider the recent COVID-19 pandemic as a motivating case for the model selection and construction. Due to the delay in onset of COVID-19 symptoms~\cite{backer2020incubation, guan2020clinical, li2020early, lauer2020incubation} and large asymptomatic populations estimated between $17-81\%$~\cite{byambasuren2020estimating, chang2020time, mizumoto2020estimating, ing2020covid}, we choose the susceptible-exposed-infected-removed (SEIR) model as the foundation of our model development.

Previous work involving the incorporation of population flows in epidemic process models include analysis of a networked SIS model with flows \cite{ifac_flow} as well as using a networked SIR model with flows to predict arrival times for various epidemics using global flight data \cite{brockmann2013hidden}, where both models are developed in continuous time. This paper uses similar derivation techniques to define our discrete-time epidemic model. 
However, we contribute to the development of such models by including the exposed state in our model formulation, as well as provide analysis of the discrete time dynamics. While other work has considered capturing the effect of transportation on the spread of COVID-19 using the SEIR model \cite{vrabac2021capturing}, the key distinction in this work is that infection propagation over the network is modeled by the relocation of infected individuals to other sub-populations rather than assuming direct contact and intermingling between sub-populations. Furthermore, while previous work assumes the likelihood of individuals traveling is independent of their infection state, this work considers the effect of being infected on the probability of an individual traveling.
% \phil{[we should also highlight that the prior work assumed everyone was equally likely to travel and how we don't]}   

\subsection{Contributions}
In this paper, we explore a discrete-time networked SEIR model to analyze the effect of population flows on the propagation of an infectious disease. We summarize the contributions of this paper as follows:

\begin{enumerate}
    \item We derive a discrete-time networked SEIR model incorporating population flows, including the conditional probabilities of a given state affecting the likelihood of an individual to be traveling.
    \item We present 
    % sufficient conditions 
    assumptions
    on the parameters for the model to be
    well defined given proper initial conditions.
    \item We prove model convergence and show that the equilibria (i.e., the healthy states) are homogeneous, that is, the same value for all sub-populations.
    \item We present a feedback control law that restricts system flow globally using the infection states, show that under the control effort the model is still well defined, and illustrate its 
    % effectiveness 
    behavior
    via simulation.
\end{enumerate}

In the remainder of this paper, we present our model derivation and necessary assumptions in Section \ref{sec:model}, provide analysis of the healthy states in Section \ref{sec:analysis}, propose a feedback control in Section \ref{sec:control}, and illustrate the model and the effectiveness of the controller through simulation in Section \ref{sec:simulations}. Finally, we provide our conclusions in Section~\ref{sec:conclusion}.

\subsection{Notation}
We denote the set of real numbers, positive real numbers, non-negative integers, and the positive integers as $\mathbb{R}$, $\mathbb{R}_{>0}$, $\mathbb{Z}_{\geq 0}$, and $\mathbb{Z}_{\geq 1}$, respectively. For any positive integer $n$, we have $[n] =\{1,2,...,n \}$. 
%We use $[0,1]^{2n}$ to denote the space of all the vectors in $\mathbb{R}^{2n}$ 
% $2\times n$ vector
%whose entries values are between 0 and 1. 
%The spectral radius of a matrix $A \in \mathbb{R}^{n\times n}$ is $ \rho(A)$. 
A diagonal matrix is denoted as diag$(\cdot)$. The transpose of a vector $x\in \mathbb{R}^n$ is $x^\top$. We use $\mathbf{0}$ and $\mathbf{1}$ to denote the vectors whose entries all equal 0 and 1, respectively. We let $\mathcal{G} =(\mathbb{V},\mathbb{E},\mathbb{W})$ denote a weighted directed graph 
% or network 
where $\mathbb{V} = \{ v_1, v_2,..., v_n\}$ is the set of nodes, $\mathbb{E} \subseteq \mathbb{V}\times \mathbb{V}$ is the set of edges, and 
$\mathbb{W}:\mathbb{E}\rightarrow \mathbb{R}_{>0}$
% $\mathbb{W} \in \mathbb{R}, e \in \mathbb{E}$ are 
maps to 
the real valued edge weights on each edge. We denote the configuration of edges in a directed graph at time step $k$ as $\mathcal{G}^k=(\mathbb{V},\mathbb{E}^k,\mathbb{W})$, where 
% $e=(j,i) \in \mathbb{E}^k$ if and only if there exists a non-zero edge weight from node $i$ to node $j$ at time $t$.
$\mathbb{E}^k$ denotes the set of edges at time step $k$. 
Furthermore, we denote $\cup_{k \geq 0}\mathbb{E}^k$ as the union of all non-zero edge configurations on a graph for all $k \geq 0$. We define a graph $\mathcal{G}$ as being strongly connected if there is a path from every node to every other
node in the graph.

\section{SEIR Model with Network Flows} \label{sec:model}
In this section, we present a networked SEIR model incorporating the population flow 
of
individuals between sub-populations. First, consider a group of $n$ sub-populations in a graph, where each sub-population $i\in [n]$ is represented by a node
% \phil{$v_i$ [is introducing this notation necessary since we it looks like don't ever use it?]} 
in the graph. We use the SEIR model to describe how suceptible individuals in sub-population $i$ become exposed, infected, and eventually recover 
% in the presence 
as the result
of an infectious disease \cite{shu2012global}. 
% \phil{states [this term is unclear... we should explain what we mean by `states'... also `transition' is close to `flow' so it's a bit unclear what the difference is between travel and epidemic spread/progression]} 
% \phil{[let's add a citation for the group SEIR model here]}
We begin with defining the SEIR model behavior without graph connections 
for each sub-population $i\in [n]$ where $S_i,E_i,I_i$, and $R_i$ represent the number of susceptible, exposed, infected, and recovered individuals in sub-population $i$, respectively, 
\begin{subequations}
\label{eq:seir_standard}
\begin{align}
        \dot{S}_i(t) &= -\beta_i(t) \frac{I_i(t)}{N_i} S_i(t) \\
        \dot{E}_i(t) &= \beta_i(t) \frac{I_i(t)}{N_i} S_i(t) - \sigma_i(t) E_i(t) \\
        \dot{I}_i(t) &= \sigma_i(t) E_i(t) - \delta_i(t) I_i(t) \\
        \dot{R}_i(t) &= \delta_i(t) I_i(t),
\end{align}
\end{subequations}
\noindent where $\beta_i$ is the infection rate, $\sigma_i$ is transition rate from exposed to infected, and $\delta_i$ is the healing rate.
% \edit{The state variables} $S_i,E_i,I_i$, and $R_i$ represent the number of susceptible, exposed, infected, and recovered individuals sub-population $i$, respectively, and
We assume $S_i(t)+E_i(t)+I_i(t)+R_i(t)=N_i$ for all $t$, i.e., 
a fixed population size 
for each sub-population. We assume fixed sub-populations as the intended time scale of the model is such that population change due to birth/death rates and migration is negligible (e.g. rapid disease propagation over weeks or months). While all the variables (and model parameters, except population), will (may) continue to vary with time, we remove the time-dependence notation for convenience and ease of reading from this point forward.
% \phil{ [this sentence is a bit long; perhaps we take the last phrase and make it its own sentence especially since it's a fairly important assumption]}. 
% \phil{[We are ignoring birth and death rates... we should either add them and show that the fixed population size implies they are equal, or we should comment as to why we are ignoring them (e.g. short time scale)]}

To account for the flow of individuals between sub-populations we 
we expand the model
in~\eqref{eq:seir_standard}:
\begin{subequations}
\label{eq:seir_flows_indv}
\begin{align}
    \dot{S}_i &= -\beta_i \frac{I_i}{N_i} S_i +\sum_{j \neq i} \left(F_{ij}\frac{S_j}{N_j}-F_{ji}\frac{S_i}{N_i} \right) \\
    \dot{E}_i &= \beta_i \frac{I_i}{N_i} S_i - \sigma_iE_i+ \sum_{j \neq i} \left(F_{ij}\frac{E_j}{N_j}-F_{ji}\frac{E_i}{N_i} \right) \\
    \dot{I}_i &= \sigma_iE_i - \delta_i I_i+ \sum_{j \neq i} \left(F_{ij}\frac{I_j}{N_j}-F_{ji}\frac{I_i}{N_i} \right) \\        
    \dot{R}_i &= \delta_i I_i + \sum_{j \neq i} \left(F_{ij}\frac{R_j}{N_j}-F_{ji}\frac{R_i}{N_i} \right),
\end{align}
\end{subequations}
% \phil{[doesn't $F_{ii}=0$? Why then do we need summation over $j \neq i$? Alternatively, we could sum over the neighborhood set of $i$, $\mathcal{N}_i$. But either way, I think it's important to point out that $F_{ii}=0$]}
% 
\noindent where $F_{ij}$ represents the number of individuals flowing from sub-population $j$  to $i$, where $F_{ii}=0$. By making a substitution of variables where $s_i = S_i/N_i, e_i = E_i/N_i, x_i = I_i/N_i, r_i = R_i/N_i$ we can model the proportion of individuals 
% in each state 
as follows
\begin{subequations}
\label{eq:seir_flows_proportion}
\begin{align}
    \dot{s}_i &= -\beta_i x_i s_i +\frac{1}{N_i}\sum_{j \neq i} \left(F_{ij}s_j-F_{ji}s_i \right) \\
    \dot{e}_i &= \beta_i x_i s_i - \sigma_i e_i+ \frac{1}{N_i}\sum_{j \neq i} \left(F_{ij}e_j-F_{ji}e_i \right) \\
    \dot{x}_i &= \sigma_i e_i - \delta_i x_i+ \frac{1}{N_i}\sum_{j \neq i} \left(F_{ij}x_j-F_{ji}x_i \right) \\        
    \dot{r}_i &= \delta_i x_i + \frac{1}{N_i}\sum_{j \neq i} \left(F_{ij}r_j-F_{ji}r_i \right), %\phil{,} %.
\end{align}
\end{subequations}

\noindent where $s_i+e_i+x_i+r_i=1$. Note that both \eqref{eq:seir_flows_indv} and \eqref{eq:seir_flows_proportion} 
assume
% assumes 
the sub-populations are well mixed and that the likelihood of an individual traveling 
% while being either 
is independent of their infectious state, that is, whether they are
susceptible, exposed, infected, or recovered. 
We now extend our model to include the probability that an individual is traveling, given their infectious state.
\begin{subequations}
\label{eq:seir_p(q|T)}
\begin{align}
    \dot{s}_i &= -\beta_i x_i s_i +\frac{1}{N_i}\sum_{j \neq i} \left(F_{ij}P(s_j|T_j)-F_{ji}P(s_i|T_i) \right) \\
    \dot{e}_i &= \beta_i x_i s_i - \sigma_i e_i+ \frac{1}{N_i}\sum_{j \neq i} \left(F_{ij}P(e_j|T_j)-F_{ji}P(e_i|T_i) \right) \\
    \dot{x}_i &= \sigma_i e_i - \delta_i x_i+ \frac{1}{N_i}\sum_{j \neq i} \left(F_{ij}P(x_j|T_j)-F_{ji}P(x_i|T_i) \right) \\        
    \dot{r}_i &= \delta_i x_i + \frac{1}{N_i}\sum_{j \neq i} \left(F_{ij}P(r_j|T_j)-F_{ji}P(r_i|T_i) \right),
\end{align}
\end{subequations}
\noindent where $P(q_i|T_i), \, q_i \in \{s_i,e_i,x_i,r_i \}$ is the probability of an individual at sub-population $i$ being in a certain infectious state given that they are also traveling. Note that
\begin{equation} \label{eq:probs_sum_to_1}
    P(s_i|T_i)+P(e_i|T_i)+P(x_i|T_i)+P(r_i|T_i) = 1
\end{equation}
\noindent for all $i \in [n]$. We can compute the probability that an individual from sub-population $i$ is traveling given that they are in state~$q_i$ 
% conditional probability 
using Bayes' rule,
\begin{equation}
\label{eq:bayes}
    P(q_i|T_i) = \frac{P(T_i|q_i)P(q_i)}{P(T_i)},
\end{equation}
\noindent where  
% $P(T_i|q_i)$ is the probability that an individual at node $i$ is traveling given that they are in state $q_i$, 
$P(T_i)$ is the probability of an individual traveling from sub-population $i$ 
% , 
and $P(q_i) = q_i$ is the probability they are in state $q_i$.
We compute the probability of an individual traveling from a given sub-population $i\in [n]$ as 
\begin{equation}
\label{eq:gamma}
    P(T_i) = \gamma_i = \frac{\sum_{j \neq i}F_{ji}}{N_i},
\end{equation}
% \phil{[doesn't $F_{ii}=0$?]}

\noindent where $\gamma_i$ is the proportion of the population flowing out of sub-population $i$ and $\gamma_i \in [0,1]$ as it is reasonable to assume that $\sum_{j \neq i}F_{ji} \ll N_i $ (as $\sum_{j \neq i}F_{ji} = N_i$ would imply that the entire population is traveling at a given time). Since measuring $P(q_i|T_i)$ is practically challenging, we instead parameterize $p_i^{q} := P(T_i|q_i)$, for $q_i \in \{s_i,e_i,x_i,r_i \}$, as follows. Using \eqref{eq:probs_sum_to_1}-\eqref{eq:gamma} we have that
\begin{align}
    % P(s_i|T_i)+P(e_i|T_i)+P(x_i|T_i)+P(r_i|T_i) &= 1 \\
    \frac{1}{\gamma_i}(p_i^{s}s_i+p_i^{e}e_i+p_i^{x}x_i+p_i^{r}r_i) &= 1. 
    % \\
\end{align}   
Therefore, assuming that $p_i^{s}= p_i^{e} = p_i^{r} = p_i^T$, we have 
\begin{align*}    
    p_i^T(s_i+e_i+r_i)+p_i^{x}x_i &= \gamma_i.
\end{align*}
\noindent Thus, solving for $p_i^T$ yields
\begin{equation}
\label{eq:p_i^T}
     p_i^T = \frac{\gamma_i -p_i^{x}x_i }{s_i+e_i+r_i},
\end{equation}
% \phil{[so does this depend on time? If so, perhaps it should be defined as $p_i^T(t)$]}
% 
\noindent 
% where we require that $p^x_i x_i \leq \gamma_i$. 
% \phil{[why do we need $p_i^x>0$? It seems bad to enforce a non-zero probability of sick people traveling... Also, is $p^x_i x_i \leq \gamma_i$ required or does it hold by construction?]}. 
% This 
which allows us to use $p_i^x \in [0,1]$ as a parameter to describe how likely an individual will be traveling given that they are infected. Furthermore, we can compute the number of individuals flowing from sub-population $j$ to $i$ as  
\begin{equation} \label{eq:F_ij}
    F_{ij} = \gamma_j w_{ij} N_j,
\end{equation}
where $w_{ij}$ is the proportion of traveling individuals flowing from sub-population $j$ to $i$ computed as
\begin{equation}\label{eq:normalized}
    w_{ij} = \frac{F_{ij}}{\sum_{l \neq j}F_{lj}}
\end{equation}
% implying that 
with $w_{ii} = 0$. % and $\sum_{j \neq i}w_{ji} = 1$.  
% [shouldn't this be $\sum_{j \neq i}w_{ij} = 1$ and doesn't it hold be construction? If not, the denominator of the previous equation is wrong. Whichever it is, this statement should be presented pointing out that it holds by construction]}. 
Thus, we can derive the dynamics for the susceptible proportion at sub-population $i$ as
\begin{align*}
    \dot{s}_i &= -\beta_i x_i s_i +\frac{1}{N_i}\sum_{j \neq i} \left(F_{ij}P(s_j|T_j)-F_{ji}P(s_i|T_i) \right) \\
     &= -\beta_i x_i s_i +\frac{1}{N_i}\sum_{j \neq i} \left(\gamma_j w_{ij} N_j\frac{p_j^T s_j}{\gamma_j}-\gamma_i w_{ji} N_i\frac{p_i^T s_i}{\gamma_i} \right) \\
     &= -\beta_i x_i s_i +\sum_{j \neq i} \left(\frac{N_j}{N_i} w_{ij}p_j^T s_j- w_{ji} p_i^T s_i \right).
\end{align*}
\noindent 
Using the fact that $\sum_{j \neq i}w_{ji} = 1$, by \eqref{eq:normalized}, we have that
\begin{equation*}
     \dot{s}_i = -(\beta_i x_i+p_i^T) s_i +\sum_{j \neq i} \frac{N_j}{N_i} w_{ij}p_j^T s_j.
\end{equation*}
\noindent By similar derivations, we can rewrite \eqref{eq:seir_p(q|T)} as
\begin{subequations}
\label{eq:flows_ind_node_cont}
\begin{align}
     \dot{s}_i &= -(\beta_i x_i+p_i^T) s_i +\sum_{j \neq i} \frac{N_j}{N_i} w_{ij}p_j^T s_j \\
     \dot{e}_i &= \beta_i x_i s_i - (\sigma_i+p_i^T)e_i + \sum_{j \neq i} \frac{N_j}{N_i} w_{ij}p_j^T e_j \\
     \dot{x}_i &= \sigma_i e_i - (\delta_i+p_i^{x})x_i + \sum_{j \neq i} \frac{N_j}{N_i} w_{ij}p_j^{x} x_j \\
     \dot{r}_i &= \delta_i x_i -p_i^{T}r_i +\sum_{j \neq i} \frac{N_j}{N_i} w_{ij}p_j^{T} r_j.
\end{align}
\end{subequations}
% \phil{[I really like this continuous-time model. I think we may be selling ourselves a bit short by just moving on from it right away...]}
% 
\noindent  We choose to discretize our model due to the nature of the collected data on the spread of pandemics, where the highest resolution data is typically recorded once per day. Using Euler's method, we can write \eqref{eq:flows_ind_node_cont} in discrete time as

\small
\begin{subequations} 
\label{eq:flows_ind_node_disc}
\begin{align}
     s_i^{k+1} &= s_i^k + h \left( -(\beta_i^k x_i^k+p_i^{T,k}) s_i^k +\sum_{j \neq i} \frac{N_j}{N_i} w_{ij}^k p_j^{T,k} s_j^k \right) \label{eq:s_flows_disc} \\
     e_i^{k+1} &= e_i^k + h \left(\beta_i^k x_i^k s_i^k - (\sigma_i^k+p_i^{T,k})e_i^k + \sum_{j \neq i} \frac{N_j}{N_i} w_{ij}^k p_j^{T,k} e_j^k \right) \\
     x_i^{k+1} &= x_i^k + h \left(\sigma_i^k e_i^k - (\delta_i^k+p_i^{x,k})x_i^k + \sum_{j \neq i} \frac{N_j}{N_i} w_{ij}^k p_j^{x,k} x_j^k \right) \\
     r_i^{k+1} &= r_i^k + h \left( \delta_i^k x_i^k -p_i^{T,k}r_i^k +\sum_{j \neq i} \frac{N_j}{N_i} w_{ij}^k p_j^{T,k} r_j^k \right),
\end{align}
\end{subequations}
\normalsize

\noindent where $k \in \mathbb{Z}_{\geq 0}$ is a given time step and $h \in \mathbb{R}_{>0}$ is a sampling parameter, yielding our discrete time model.

For the model in \eqref{eq:flows_ind_node_disc} to be well-defined we require the following assumptions.
\begin{assumption}
\label{assume:equal_flows}
Let $\sum_{i \neq j} F_{ij}^k = \sum_{i \neq j}F_{ji}^k$ for all $i \in [n]$ and $k \in \mathbb{Z}_{\geq 0}$.
\end{assumption}
% First, we assume that \phil{$\sum_{j \neq i} F_{ji} = \sum_{j \neq i}F_{ij}$ [why not include this in Assumption \ref{assume}? Or, maybe even better, it should have it's own assumption environment]} for all $i \in [n]$,
\noindent This assumption requires that the total flow of individuals into a given sub-population must be equal to the total flow out. 
% \phil{This condition follows naturally from the assumption that the sub-population sizes remain constant [I think this is backwards... I believe Assumption \ref{assume:equal_flows} is necessary for the sub-populations' sizes to be fixed, right?]}, which is reasonable when modeling rapidly spreading diseases over a shorter time scale. 
% \phil{[maybe comment on time scale here? That is, we interested in a shorter time scales where migration will be non-consequential or something like that]} 
Furthermore, we impose the following assumption on the model parameters.

\begin{assumption}
\label{assume:parameters}
For all $i\in [n]$ and $k \in \mathbb{Z}_{\geq 0}$, let $\beta_i^k,\delta_i^k,\sigma_i^k \in \mathbb{R}_{>0}$, $h\beta_i^k, h\delta_i^k, h \sigma_i^k \in (0,1]$, $h(\beta_i^k+p^{T,k}_i) \leq 1, h(\sigma_i^k+p^{T,k}_i) \leq 1,$ and $h(\delta_i^k+p^{T,k}_i) \leq 1$. 
% \ciyuan{We might need the following assumptions: nonzero initial conditions and the graph G is strongly connected.}
% \phil{[but they should probably be separate assumptions (or a single extra assumption) after Lemma \ref{lem:well_defined}]}
\end{assumption}

Under these assumptions, we can show that given proper initial conditions the model will always remain well defined.

\begin{lemma}
\label{lem:well_defined}
Consider the model in \eqref{eq:flows_ind_node_disc} under Assumptions \ref{assume:equal_flows}-\ref{assume:parameters}. Suppose $s_i^0,e_i^0,x_i^0,r_i^0 \in [0,1]$ and $s_i^0+e_i^0+x_i^0+r_i^0=1$ for all $i\in [n]$. Then, for all $k \geq 0$ and $i\in [n]$, $s_i^k,e_i^k,x_i^k,r_i^k \in [0,1]$ and $s_i^k+e_i^k+x_i^k+r_i^k=1$.
\end{lemma}

\begin{proof}
We prove the result by induction. By assumption, it holds for the base case $k=0$. We now show the inductive step, where given $s_i^k,e_i^k,x_i^k,r_i^k \in [0,1]$ and $s_i^k+e_i^k+x_i^k+r_i^k=1$ for all $i\in [n]$, we show that the same holds for $k+1$. By \eqref{eq:flows_ind_node_disc} and Assumption \ref{assume:parameters} we  have 
\begin{align}\label{eq:nonneg}
    q_i^{k+1} & \geq h\sum_{j \neq i} \frac{N_j}{N_i} w_{ij}^k p_j^{T,k} q_j^k \geq 0,
\end{align}
where $q_i^k \in \{s^k_i,e^k_i,x^k_i,r^k_i\}$ for all $i\in [n]$ and $k \in \mathbb{Z}_{\geq 0}$. 

Furthermore, computing $s_i^{k+1}+e_i^{k+1}+x_i^{k+1}+r_i^{k+1}$ yields
\small
\begin{align}\nonumber 
    & s_i^{k+1}+e_i^{k+1}+x_i^{k+1}+r_i^{k+1} = s_i^k+e_i^k+x_i^k+r_i^k 
    % \\
    % & 
    + h \left(-\beta_i^k x_i^k s_i + \beta_i^k x_i^k s_i - \sigma_i^k e_i^k +\sigma_i^k e_i^k - \delta_i^k x_i^k +\delta_i^k x_i^k \right) \nonumber\\
    & + \frac{h}{N_i}  \sum_{j \neq i} \Big( F_{ij}^k(P(s_i|T_i)+P(e_i|T_i)+P(x_i|T_i)+P(r_i|T_i)) \nonumber 
    % \\  &
    - F_{ji}^k(P(s_j|T_j)+P(e_j|T_j)+P(x_j|T_j)+P(r_j|T_j)) \Big) \nonumber \\
    % s_i^{k+1}+e_i^{k+1}+x_i^{k+1}+r_i^{k+1} 
    & \ \ \ \ \ \ \ \ \ \ \ \ \ \ \ \ \ \ \ \ \ \ \ \ \ \ \ \ \ \ =s_i^k+e_i^k+x_i^k+r_i^k + \frac{h}{N_i} \left(F^+ - F^- \right) = 1, \label{eq:sum1}
\end{align}
\normalsize
\noindent where $F^+ = \sum_{j \neq i} F_{ij}^k$ and $F^- = \sum_{j \neq i} F_{ji}^k$ denote the total flow in and out of each sub-population, respectively, which are equal by Assumption \ref{assume:equal_flows}. Since
each variable is non-negative by \eqref{eq:nonneg} and their sum must add to 
one by \eqref{eq:sum1}, we have that $s_i^{k+1},e_i^{k+1},x_i^{k+1},r_i^{k+1} \leq 1$. 
Thus, we have shown
that $s_i^{k+1},e_i^{k+1},x_i^{k+1},r_i^{k+1} \in [0,1]$, completing the inductive step.
\end{proof}
\begin{remark}
Assumption \ref{assume:parameters} requires that the sampling parameter be small enough in relation to the model spread parameters such that the model remains well defined. Furthermore, requiring $h(\beta_i^k x_i^k+p^{T,k}_i) \leq 1, h(\sigma_i^k+p^{T,k}_i) \leq 1,$ and $h(\delta_i^k+p^{T,k}_i) \leq 1$ can be interpreted as requiring
that no individual can both travel and transition between infectious
states during the same time step $k \in \mathbb{Z}_{\geq0}$, as our model does not capture infection occurring 
during travel.
\end{remark}
The following are not required for the model to remain well defined. However, we use them in the next section to show that the set of healthy states have a homogeneous structure.

\begin{definition}
A graph $\mathcal{G}^k= (\mathbb{V},\cup_{k\geq 0}\mathbb{E}^k,\mathbb{W})$ for $k \in \mathbb{Z}_{\geq 0}$ is \emph{$K$-strongly connected} if there exist some bound $K$ such that   
$(\mathbb{V},\cup_{j=k}^{k+K-1} \mathbb{E}^j, \mathbb{W})$
is strongly connected, for all $k \in \mathbb{Z}_{\geq 0}$.
\end{definition}
\begin{assumption} \label{assume:strongly_connected_bounded_comms}
Let the graph $\mathcal{G}^k = (\mathbb{V},\cup_{k\geq 0}\mathbb{E}^k,\mathbb{W})$, 
% \phil{[add a weighting part to the graph notation?]}, 
where 
% the weighted edges of the graph from node $j$ to $i$ at time step $k$ are defined by 
$\mathbb{W}:\mathbb{E}^k\rightarrow \mathbb{R}_{>0}$ is defined by
$w_{ij}^k$, be $K$-strongly connected.
\end{assumption}
%%%%%%%%%%%%%%%%%%%%%%%%%%%%%%%%%%%%%%%%%%%%%%%%%%%%%%%%%%%%%%%%%%%%%

\section{Model Analysis} \label{sec:analysis}
In this section we analyze the equilibria of the model in \eqref{eq:flows_ind_node_disc}, i.e., the healthy states of the system, which we define as $q_i^* = \lim_{k \rightarrow \infty} q_i^k$  for all $i \in [n]$ where $q_i^* \in \{s_i^*,e_i^*,x_i^*,r_i^*\}$. We use the following result given in \cite{blondel2005convergence} on the conditions for discrete-time consensus models.

\begin{lemma} \label{lem:consensus}
Let a discrete-time system defined by the transition matrix $L^k$ satisfy the following properties, where $l_{ij}^k$ is the corresponding entry in the $i$th row and $j$th column at time step $k \in \mathbb{Z}_{\geq 0}$:
\begin{enumerate}[label=(\roman*)]
    \item The graph $G = (\mathbb{V},\cup_{k\geq 0}\mathbb{E}^k)$, where the edge weights at time step $k$ are given by $L^k$, 
    is $K$-strongly connected.
    % \item The system has bounded intercommunication intervals \cite{tseng1990partially}. 
    % \phil{[This should be explained in the context of our model... it still doesn't make sense right now]}
    \item There exists a positive constant $y \in \mathbb{R}_{>0}$ such that for all $i,j \in [n]$ and $k \in \mathbb{Z}_{\geq 0}$ we have 
    \begin{enumerate}[label=(\alph*)]
        \item $l_{ii}^k \geq y$ 
        \item $l_{ij}^k \in \{0\} \cup [y,1]$ 
        % \phil{[does $[y,1]$ represent the real numbers between $y$ and 1? We should probably clarify this and/or add it to the Notation Section]}
        \item $\sum_{j=1}^n l_{ij}^k = 1$.
    \end{enumerate}
\end{enumerate}
\noindent Then, the system dynamics defined by $L^k$ guarantee asymptotic consensus.
\end{lemma}
\noindent We now present results on the asymptotic convergence of the healthy states for the system in \eqref{eq:flows_ind_node_disc}.
\begin{theorem}\label{thm:equilibrium}
Consider the model in \eqref{eq:flows_ind_node_disc} under Assumptions \ref{assume:equal_flows}-\ref{assume:strongly_connected_bounded_comms}. Given that there exists some $i \in [n]$ such that $x_i^0\in (0,1]$ or $e_i^0\in (0,1]$, %\ciyuan{and $\max_{i\in [n]}{\beta_i} \leq \min_{i \in [n]}{\sigma_i}/n$}, 
then there exists a set of asymptotically stable equilibria of the form $(\mathbf{s}^*,\mathbf{0},\mathbf{0},1-\mathbf{s}^*)$, 
% is the unique 
where $\mathbf{s}^* = \alpha\mathbf{1}$, $\alpha \in [0,1]$. 
% \ciyuan{and the healthy state is asymptotically stable.}
% , and $s_i^* = \sigma$ for all $i \in [n]$.

% % a
% {\color{blue}the}
% unique equilibrium.
% {\color{blue}[it seems like the whole theorem should just be part 3) and maybe you split the proof into these three parts, if needed... and perhaps we introduce the $q_i^*$ notation before the theorem statement... also we need some assumption that at least someone is sick or exposed and that the network is (strongly) connected (counter example: think a directed tree) for this result to be true... so maybe add Assumptions 2 and 3]}
\end{theorem}

\begin{proof}
We prove the result by splitting it into
two parts, namely,
\begin{enumerate}
    \item $e_i^* = 0$ and $x_i^* = 0$, for all $i \in [n]$
    % \item $(\mathbf{s}^*,\mathbf{0},\mathbf{0},1-\mathbf{s}^*)$ is unique
    \item $s_i^* = \alpha
    $  for all $i \in [n]$.
\end{enumerate}
The parts are presented sequentially:

1) 
We denote $\mathbf{S}^k =\sum_{i=1}^n s_i^k \in [0,n]$, $\mathbf{E}^k =\sum_{i=1}^n e_i^k\in [0,n]$ and $\mathbf{X}^k = \sum_{i=1}^n x_i^k\in [0,n]$ as the sum of the susceptible, exposed, and infected states of all the sub-populations, respectively. Hence, based on the dynamics of the system states in~\eqref{eq:flows_ind_node_disc} and by Assumption \ref{assume:equal_flows}, we have
\begin{subequations}
 \label{eq:flows_ind_node_disc_sum_simplified}
\begin{align}
     \mathbf{S}^{k+1} &=\mathbf{S}^k - h \sum_{i=1}^n\beta_i^k x_i^k s_i^k \\
     \mathbf{E}^{k+1} &=\mathbf{E}^k + h \left(\sum_{i=1}^n\beta_i^k x_i^k s_i^k - \sum_{i=1}^n \sigma_i^ke_i^k \right) \\
     \mathbf{X}^{k+1} &= \mathbf{X}^k + h \left(\sum_{i=1}^n \sigma_i^k e_i^k -\sum_{i=1}^n \delta_i^k x_i^k   \right).
\end{align}
\end{subequations}
By Assumption~\ref{assume:parameters} and Lemma~\ref{lem:well_defined}, we have that the rate of change of $\mathbf{S}^k$, $-h \sum_{i=1}^n\beta_i^k x_i^k s_i^k$, is non-positive for all $k\geq 0$ and $\mathbf{S}^k$ is lower bounded by zero. Hence, by Lemma~\ref{lem:well_defined}, we obtain that $\lim_{k\rightarrow \infty}\mathbf{S}^k$ exists, and
\begin{equation}
    \lim_{k\rightarrow \infty} -h \sum_{i=1}^n\beta_i^k x_i^k s_i^k=0.
\end{equation}
Accordingly, we can write that $\lim_{k\rightarrow \infty}(\mathbf{E}^{k+1}-\mathbf{E}^{k})=\lim_{k\rightarrow \infty}-h\sum_{i=1}^n \sigma_i^ke_i^k$. 
From Assumption~\ref{assume:parameters}, we know that $h\sigma_i^k>0$ for all $i\in[n]$ and $k\in \mathbb{Z}_{\geq 0}$. Therefore, by Lemma~\ref{lem:well_defined}, we conclude that $\lim_{k\rightarrow \infty}\mathbf{E}^{k}=0$ and thus $\lim_{k\rightarrow 
\infty}e_i^{k}=0$ for all $i\in[n]$.

Similarly, we acquire that $\lim_{k\rightarrow \infty}(\mathbf{X}^{k+1}-\mathbf{X}^{k})=\lim_{k\rightarrow \infty}-h\sum_{i=1}^n \delta_i^kx_i^k$, since $\lim_{k\rightarrow \infty}e_i^{k}=0$ for all $i\in[n]$. By Assumption~\ref{assume:parameters} and Lemma~\ref{lem:well_defined}, $h\delta_i^k>0$ and $x_i^k$ is well-defined, we acquire that $\lim_{k\rightarrow \infty}\mathbf{X}^{k}=0$ and thus $\lim_{k\rightarrow \infty}x_i^{k}=0$ for all $i\in[n]$.

2) In order to show that the susceptible states become equally mixed
as $k \rightarrow \infty$, where $k \in \mathbb{Z}_{\geq 0}$, we leverage Lemma~\ref{lem:consensus}.
First, we construct dynamics for the susceptible states in \eqref{eq:s_flows_disc} when $x_i^*=e_i^* = 0$ for all $i \in [n]$, as proven in 1), which yields
\begin{equation} \label{eq:suseptible-x_i=0}
    s_i^{k+1} = s_i^k + h \left( - \gamma_i^k s_i^k +\sum_{j \neq i} \frac{N_j}{N_i} w_{ij}^k \gamma_i^k s_j^k \right),
\end{equation}
\noindent where $p^{T,k}_i = \gamma_i$ by \eqref{eq:p_i^T} and Lemma \ref{lem:well_defined}. We can express these dynamics in matrix form as 
\begin{equation} \label{eq:suseptible-x_i=0-matrix}
    s^{k+1} = \underbrace{\left( I + h \left( -\Gamma^k + N^{-1}W^k \Gamma^k N \right) \right)}_{L^k} s^k,
\end{equation}
\noindent where $s^k = [s_1^k, \dots, s_n^k]^{\top}$, $\Gamma^k = \text{diag}( \gamma_1^k ,\dots,\gamma_n^k)$, $N = \text{diag}(N_1,\dots,N_n)$, and $W^k$ is the matrix defined by entries $w_{ij}^k$ on the $i$th row and $j$th column. In order to guarantee that $s_i^*$ is homogeneous for all $i \in [n]$, 
we will apply Lemma~\ref{lem:consensus} to~\eqref{eq:suseptible-x_i=0-matrix}.

Property \textit{(i)} of 
Lemma~\ref{lem:consensus} holds
by Assumption~\ref{assume:strongly_connected_bounded_comms}. Furthermore, \textit{(ii.a)} and \textit{(ii.b)} are true by construction. We now show \textit{(ii.c)} is true by computing $L^k \mathbf{1}$, which yields
\begin{align}
    &L^k \mathbf{1} = \mathbf{1} + h\left( -\bm{\gamma}^k + N^{-1}W^k \Gamma^k N \mathbf{1} \right),
\end{align}
\noindent where $\bm{\gamma}^k = [\gamma_i^k, \dots, \gamma_n^k]^\top$. Computing $N^{-1}W^k \Gamma^k N \mathbf{1}$ yields
\begin{align} \label{eq:L1_first_step}
    &N^{-1}W^k \Gamma^k N \mathbf{1} = N^{-1}
    \begin{bmatrix}
    \sum_{j \neq 1}\gamma_j^k w_{1j}^k N_j \\ \vdots \\ \sum_{j \neq n}\gamma_j^k w_{nj}^k N_j 
    \end{bmatrix}. 
\end{align}
\noindent Using \eqref{eq:F_ij}, we can write \eqref{eq:L1_first_step} as 
\begin{align}
    N^{-1}
    \begin{bmatrix}
    \sum_{j \neq 1} F_{1j}^k \\ \vdots \\ \sum_{j \neq n}F_{nj}^k
    \end{bmatrix} = 
    \begin{bmatrix}
    \sum_{j \neq 1} \frac{F_{1j}^k}{N_1} \\ \vdots \\ \sum_{j \neq n}\frac{F_{nj}^k}{N_n}
    \end{bmatrix} ,
\end{align}
\noindent which, by Assumption \ref{assume:equal_flows} and \eqref{eq:gamma}, becomes
\begin{align*}
    \begin{bmatrix}
    \sum_{i \neq 1} \frac{F_{i1}^k}{N_1} \\ \vdots \\ \sum_{i \neq n}\frac{F_{in}^k}{N_n}
    \end{bmatrix} =
    \begin{bmatrix}
    \gamma_1^k \\ \vdots \\ \gamma_n^k
    \end{bmatrix} = \bm{\gamma}^k.
\end{align*}
\noindent Thus, $L^k \mathbf{1} = \mathbf{1}$, showing that $L^k$ meets the requirements to guarantee consensus of the susceptible states as stated in Lemma~\ref{lem:consensus}.  Therefore, 
% ensuring that 
$s_i^* = \alpha$ for all $i \in [n]$, where by Lemma \ref{lem:well_defined} we have that $\alpha \in [0,1]$.
\end{proof}

\section{Feedback Control} \label{sec:control}

% \phil{[This subsection doesn't seem to belong in the `Model Analysis' Section]}

We now propose a feedback control strategy for the model in \eqref{eq:flows_ind_node_disc} derived from the current infection states. Given that $\gamma_i^k$ represents the unimpeded flow out of sub-population $i$ at time step $k \in \mathbb{Z}_{\geq 0}$, we can implement a scheme that restricts travel between all sub-populations proportionally:
\begin{equation}\label{eq:control}
    \tilde{\gamma}_i^{k} = \theta^k\gamma_i^{k}, 
\end{equation}
% \phil{[I think we should have a discussion about how this model is still well defined under the control design... right now we're glossing over the nuances about flow in = flow out... we should add an assumption about the flow and have a result here that says it still holds under the proposed control strategy]}
% 
\noindent where 
$\theta^k \in [0,1]$ is the flow restriction penalty.
We propose a flow restriction penalty that is a function of the average infection level:
\begin{equation}\label{eq:theta}
    \theta^k = 1-(\bar{x}^{k})^{\frac{1}{\eta}},
\end{equation}
where
$\bar{x}^{k} = \frac{1}{n} \sum_{i \in [n]} x_i^k$ is the average proportion of infected individuals across all sub-populations and $\eta \in \mathbb{R}_{>0}$ can be viewed as a sensitivity parameter, where $\eta >1$ denotes a higher sensitivity and $\eta < 1$ denotes a lower sensitivity to the average infection level in the network.
% , respectively. 
The magnitude of $\eta$ can also be viewed as the strength of the controller in reaction to the overall infection. 
We now show that applying the strategy in \eqref{eq:control} still maintains the assumptions imposed in Section \ref{sec:model}, 
% and consequently 
enforcing that the model remains well defined.
We define $\widetilde{\eqref{eq:flows_ind_node_disc}}$ as the system with dynamics in \eqref{eq:flows_ind_node_disc} including the control strategy in \eqref{eq:control}.
% \begin{proposition}
% Consider the model in \eqref{eq:flows_ind_node_disc} with the control strategy in \eqref{eq:control}, defined as $\widetilde{\eqref{eq:flows_ind_node_disc}}$. Given that \eqref{eq:flows_ind_node_disc} is well defined and that $\theta^k \in [0,1] $ for all $ k$,  $\widetilde{\eqref{eq:flows_ind_node_disc}}$ must also remain well defined.
% \end{proposition}
\begin{proposition}\label{prop}
Consider $\widetilde{\eqref{eq:flows_ind_node_disc}}$ 
under Assumption \ref{assume:parameters} and with $\theta^k \in [0,1] $ for all $k\geq 0$. If with $s_i^0,e_i^0,x_i^0,r_i^0 \in [0,1]$ and $s_i^0+e_i^0+x_i^0+r_i^0=1$ for all $i\in [n]$, then $s_i^k,e_i^k,x_i^k,r_i^k \in [0,1]$ and $s_i^k+e_i^k+x_i^k+r_i^k=1$, for all $k \geq 0$ and $i\in [n]$. 
% Given that \eqref{eq:flows_ind_node_disc} is well defined and that $\theta^k \in [0,1] $ for all $ k$,  $\widetilde{\eqref{eq:flows_ind_node_disc}}$ must also remain well defined.
\end{proposition}
\begin{proof}
As control is only applied to $\gamma_i^k$ for all $i \in [n]$, we must only verify that Assumption \ref{assume:equal_flows} is met under the proposed control strategy and the rest will follow from Lemma \ref{lem:well_defined}. Since $\theta^k \in [0,1]$ we have by \eqref{eq:gamma} that $\tilde{\gamma}_i^{k} \in [0,1]$. 
% \phil{$\tilde{\gamma}_i^{k} \in [0,1]$ [this is only true if ${\gamma}_i^{k} \in [0,1]$... so we need to show it or point to where we show it/prove it earlier]} 
Furthermore, by \eqref{eq:F_ij} we have
\begin{equation}
    \tilde{F}_{ij}^k = \tilde{\gamma_j}^k w_{ij}^k N_j.
\end{equation}
\noindent Computing the flows in versus the flows out, $\tilde{F}^+-\tilde{F}^-$, yields
\begin{align*}
    \tilde{F}^+-\tilde{F}^- &= \sum_{j \neq i} \tilde{F}_{ij}^k - \sum_{j \neq i} \tilde{F}_{ji}^k \\
    &=\sum_{j \neq i} \tilde{\gamma_j}^k w_{ij}^k N_j - \sum_{j \neq i} \tilde{\gamma_i}^k w_{ji}^k N_i \\ 
    &= \theta^k \left(\sum_{j \neq i} \gamma_j^k w_{ij}^k N_j - \sum_{j \neq i} \gamma_i^k w_{ji}^k N_i \right) \\
    &= \theta^k(F^+-F^-) = 0.
\end{align*}
Thus, Assumption \ref{assume:equal_flows} is maintained. Therefore, by Lemma \ref{lem:well_defined} $s_i^k,e_i^k,x_i^k,r_i^k \in [0,1]$ and $s_i^k+e_i^k+x_i^k+r_i^k=1$, for all $k \geq 0$ and $i\in [n]$. 
% and the model $\widetilde{\eqref{eq:flows_ind_node_disc}}$ remains well defined.
\end{proof}

% This control strategy effectively restricts flow between all sub-populations proportionally in the presence of any amount of system infection. 
By restricting $\gamma_i^k$, we directly reduce the flow of both infected and non-infected individuals according to \eqref{eq:p_i^T} as
\begin{equation}
\small
    \theta^k \left(p_i^{T,k}(s_i^k+e_i^k+r_i^k)+p_i^{x,k}x_i^k \right) = \theta^k\gamma_i^k .
\end{equation}
\noindent Thus, our controlled flow rates, with respect to the conditional probability parameters, are given by
\begin{equation}
    \tilde{p}_i^{q,k} = \theta^k {p}_i^{q,k}  
\end{equation}
\noindent where $p_i^{q,k} \in \{p_i^{s,k},p_i^{e,k},p_i^{x,k},p_i^{r,k} \}$.

The control strategy in \eqref{eq:control}-\eqref{eq:theta} effectively restricts flow between all sub-populations proportionally in the presence of any amount of system infection. Note that by construction $ \theta^k = 1-(\bar{x}^{k})^{\frac{1}{\eta}} \in [0,1] $ for all $k\geq 0$. Thus, by Proposition~\ref{prop}, the system is well defined.

It should be noted that while proportional restrictions to all flow are not the most precise form of control that can be applied to this model, this approach is not dissimilar to the travel policies on global and regional flights during the height of the COVID-19 pandemic~\cite{suzumura2020impact}. In the following section we apply this control strategy to a simplified model of a travel network between populous cities based on median flight data, and evaluate its effectiveness on mitigating disease spread.

\section{Simulations} \label{sec:simulations}

\graphicspath{{Images/}}

In this section, we detail the methods and parameters used to simulate our model, its limiting behavior, and our proposed control strategy as described in Sections \ref{sec:model}-\ref{sec:control}. We construct our simulations using population data from the US cities of Atlanta, Los Angeles, Chicago, and Dallas, and the flights between each city's primary airport (ATL, LAX, ORD, DFW). The infection starts in Los Angeles and propagates through the network, reaching an equilibrium where $x^k_i = 0$ for all $ i \in [n]$.

To simulate the states for the SEIR model we use \eqref{eq:flows_ind_node_disc} with fixed  homogeneous spread parameters (i.e., the same for every sub-population 
and static), $(\beta,\delta,\sigma,h,p^x)=(0.5,0.34,0.19,0.14,0.005)$.
The population of each city is given by $(N_{\text{ATL}}, N_{\text{LAX}}, N_{\text{ORD}}, N_{\text{DFW}}) \approx (0.5, 4,2.7,1.3)*10^6$, where the population sizes are approximated from \cite{census2019} and \cite{datacommons}.
The population traveling between the cities is approximated by the median number of daily flights between the airports in March 2021\cite{flightsfrom}:
\begin{equation*}
    F = \xi \begin{bmatrix}
    0  & 15 & 23 & 19\\
    15 & 0  & 22 & 21\\
    23 & 22 & 0  & 23\\
    19 & 21 & 23 & 0
    \end{bmatrix},
\end{equation*}
where $\xi \in \mathbb{R}_{\geq 0}$ is a scaling factor, which is
used to increase or decrease the total volume of population flow.
The initial conditions for the sub-populations 
% distributions of the nodes 
are $s^0=[1,0.99,1,1], \ e^0=[0,0.005,0,0], \ x^0=[0,0.005,0,0]$, and $r^0=[0,0,0,0]$.

\begin{figure}
    \centering
    \begin{overpic}[trim = 0.69cm 0.8cm 0 0, clip, width = .5\columnwidth]{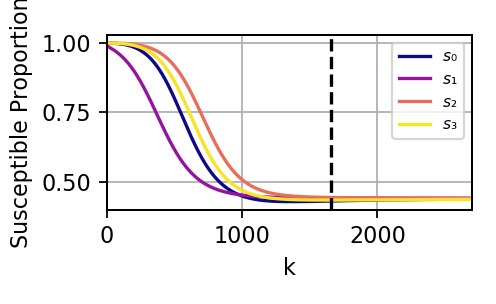}
    \put(-6.5,25){\large{\parbox{0.75\linewidth}\large \rotatebox{90}{$s_i^k$}}}
     \put(55,-2){\large{\parbox{0.75\linewidth}\large $k$
     }}\normalsize
     \put(86.5,35){\colorbox{white}{\parbox{0.03\linewidth}{\footnotesize $s_0^k$
     
     \vspace{.45ex}
     
     $s_1^k$
     
     \vspace{.45ex}
     
     $s_2^k$
     
     \vspace{.45ex}
     
     $s_3^k$}}}
    %  \put(86.5,36){\colorbox{white}{\parbox{0.03\linewidth}{\footnotesize $s_0^k$}}}
   \end{overpic}
    \caption{
    % \phil{[the border on this figure is too thick]} 
    An illustration of the model in \eqref{eq:flows_ind_node_disc} reaching consensus after the infection dies out in the system, indicated by the vertical black dotted line where $\bar{x} \leq 10^{-4}$.}
    \label{fig:consensus_plot}
\end{figure}

First, we show the consensus behavior from Theorem~\ref{thm:equilibrium} by simulating \eqref{eq:flows_ind_node_disc} using the given initial conditions and model parameters. In Figure~\ref{fig:consensus_plot}, we see that the susceptible proportions of sub-populations 
% of the nodes
reach consensus. 
The dashed vertical line in the graph indicates the time step at which the average infected proportion of the population reaches near zero ($\bar{x}^k \leq 10^{-4}$). 
% , with $\bar{x}^{k} = \frac{1}{n} \sum_{i \in [n]} x_i^k$). 
These results suggest that, by the end of the epidemic process, the flow dynamics dominate the behavior of the model and eventually lead to an equal mixing of susceptible and recovered populations given enough time, as suggested by Theorem~\ref{thm:equilibrium}.

\begin{figure}
\centering
    \begin{overpic}[trim=2cm 2.2cm 2.1cm 1,clip,width=0.48\columnwidth]{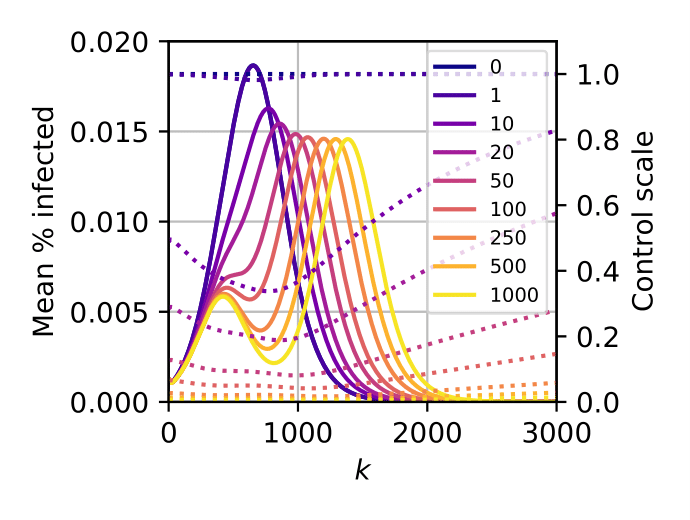}
    \put(-4,38){{\parbox{0.75\linewidth}\footnotesize \rotatebox{90}{$\bar{x}^k$}}}
    \put(50,-3){\large{\parbox{0.75\linewidth}\large $k$}}
    \put(100,38){\normalsize{\parbox{0.75\linewidth}\small \rotatebox{90}{$\theta^k$}}}\normalsize
    \end{overpic}
    \hfill
    \begin{overpic}[trim=1 2cm 1 -1cm,clip,width=0.47\columnwidth]{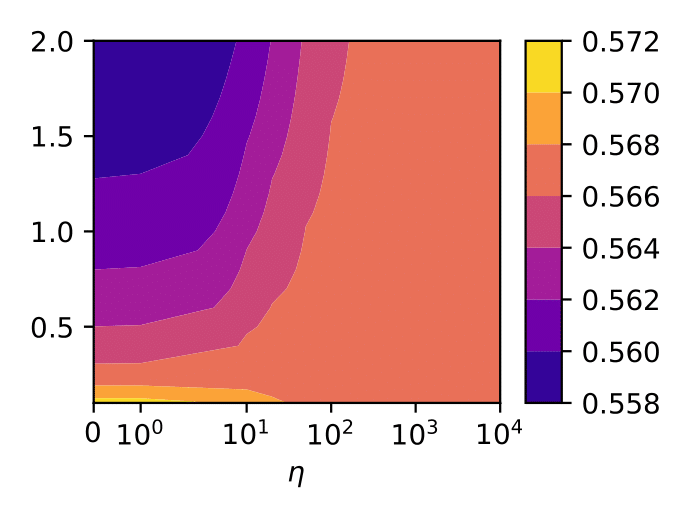}
    \put(-2,34){{\parbox{0.75\linewidth}\small \rotatebox{90}{$\gamma$}
     }}\normalsize
    \put(44,-1){{\parbox{0.75\linewidth}\small
     $\eta$ }
     }
     \put(100,34){{\parbox{0.75\linewidth}\small \rotatebox{90}{$\bar{r}^*$}}}
     \normalsize
   \end{overpic}
\caption{System response to the controller in \eqref{eq:control}-\eqref{eq:theta} with different 
strength values ($\eta$) and flow rates ($\gamma$). (Top) Plot of $\bar{x}^k$
with $\eta$ ranging from $0$ to $1000$ and $\xi=100$. An $\eta=0$ means the control strategy is not used. The dotted lines of corresponding color denote the control penalty applied
to the flow rates in the system with respect to the total infection level. (Bottom) The average proportion of recovered individuals $\bar{r}^*$ for a spectrum of equilibria.
The baseline $\gamma=1.0$ is when $\xi=100$.}
\label{fig:control_no_vaccine}
\end{figure}

\begin{figure}
\centering
    \begin{overpic}[trim=2cm 2.2cm 2.1cm 1,clip,width=0.47\columnwidth]{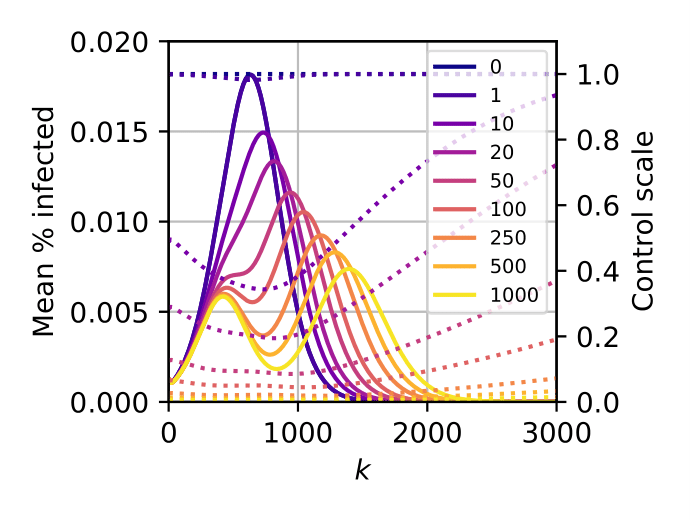}
    \put(-4,38){{\parbox{0.75\linewidth}\footnotesize \rotatebox{90}{$\bar{x}^k$}}}
    \put(50,-3){\large{\parbox{0.75\linewidth}\large $k$}}
    \put(100,38){\normalsize{\parbox{0.75\linewidth}\small \rotatebox{90}{$\theta^k$}}}\normalsize
    \end{overpic}
    \hfill
    \begin{overpic}[trim=1 2cm 1 -1cm,clip,width=0.47\columnwidth]{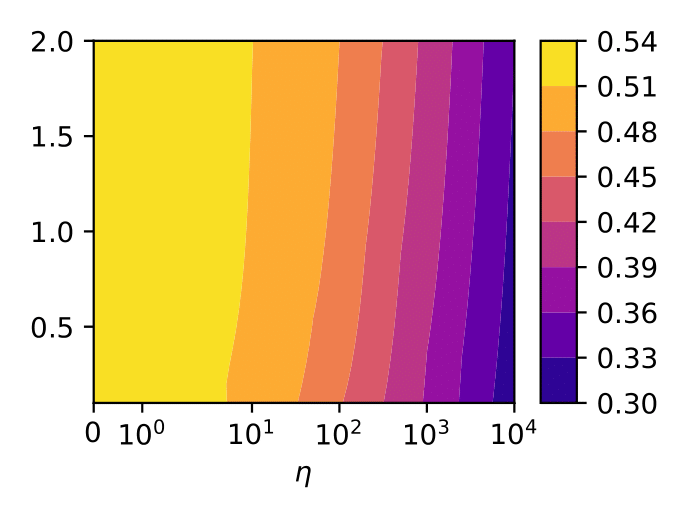}
    \put(-2,34){{\parbox{0.75\linewidth}\small \rotatebox{90}{$\gamma$}
     }}\normalsize
    \put(44,-1){{\parbox{0.75\linewidth}\small
     $\eta$ }
     }
     \put(100,30){{\parbox{0.75\linewidth}\small \rotatebox{90}{$\delta \sum_{k} \bar{x}^k$}}}
     \normalsize
   \end{overpic}
\caption{Implementation of a vaccine roll-out. Plots mirror those in Figure~\ref{fig:control_no_vaccine}, with the roll-out starting at $k=500$, moving $0.1\%$ of 
% each node's total population directly from 
$s_i^k\rightarrow r_i^k$ until $\bar{s}^k=0.01$. 
(Top) Plot of $\bar{x}^k$ 
with $\eta$ ranging from $0$ to $1000$ and $\xi=100$. An $\eta=0$ means the control strategy is not used. The dotted lines of corresponding color denote the control penalty applied to the flow rates in the system with respect to the total infection level. (Bottom) The average proportion of recovered individuals due to infection $\delta \sum_{k} \bar{x}^k$, which is equivalent to $\bar{r}^*$ when there is no vaccine,
for a spectrum of equilibria.
The baseline $\gamma=1.0$ is when $\xi=100$. }
\label{fig:control_vaccine}
\end{figure}

\begin{figure}
    \centering
    \begin{overpic}[trim=2cm 2.3cm 2.2cm 1,clip,width=0.5\columnwidth]{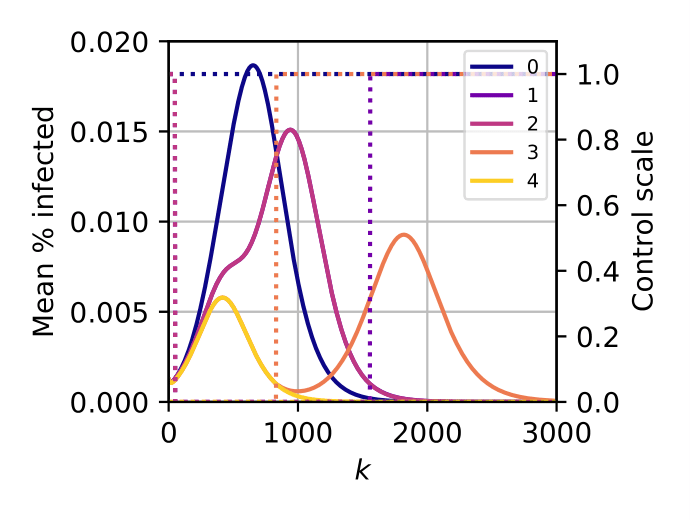}
    \put(-4,38){{\parbox{0.75\linewidth}\footnotesize \rotatebox{90}{$\bar{x}^k$}}}
    \put(50,-3){\large{\parbox{0.75\linewidth}\large $k$}}
    \put(100,38){{\parbox{0.75\linewidth}\small \rotatebox{90}{$\theta^k$}}}
    \normalsize
    \end{overpic}
    \caption{
    System response (solid lines) to the controller in \eqref{eq:control} with an on-off flow penalty (dotted lines).
    The case $\eta = 0$ has $\theta^k = 1 $ for all $ k \geq 0$, the cases $\eta = 1,2$ have $\theta^k=1 $ for all $ k \in \{0, \dots, 49\}$ and set $\theta^k=0$ starting at $k=50$, the cases $\eta = 3,4$ have $\theta^k=0$ starting from $k=0$,
    cases $\eta = 1,3$ re-open travel (set $\theta^k=1$) for all $k$ after $\bar{x}^k < 0.001$, and cases $\eta = 2,4$ only re-open travel after $\bar{x}^k \approx 0$. 
    % \edit{The dotted lines of corresponding color denote the control penalty applied at time-step $k$.}
    }
    \label{fig:binary_controller_no_vaccine}
\end{figure}

\begin{figure}
    \centering
    \begin{overpic}[trim=2cm 2.3cm 2.2cm 1,clip,width=0.5\columnwidth]{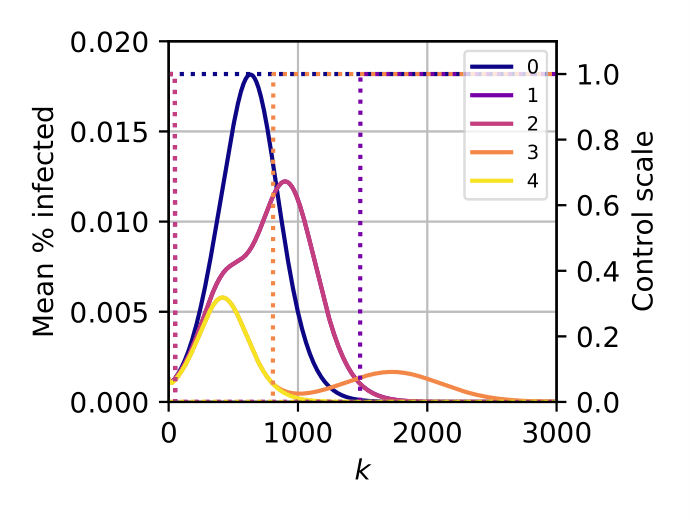}
    \put(-4,38){{\parbox{0.75\linewidth}\footnotesize \rotatebox{90}{$\bar{x}^k$}}}
    \put(50,-3){\large{\parbox{0.75\linewidth}\large $k$}}
    \put(100,38){{\parbox{0.75\linewidth}\small \rotatebox{90}{$\theta^k$}}}
    \normalsize
    \end{overpic}
    \caption{
    System response to the controller in \eqref{eq:control} with an on-off flow penalty combined with a vaccine roll-out. 
    The on-off controller is identical to  Figure~\ref{fig:binary_controller_no_vaccine} and the vaccine roll-out is identical Figure \ref{fig:control_vaccine}. 
    % \edit{The dotted lines of corresponding color denote the control penalty applied at time-step $k$.} 
    }
    \label{fig:binary_controller_vaccine}
\end{figure}

Second, we simulate the control strategy proposed in Section \ref{sec:control} with and without an additional heuristic for vaccine distribution. With no vaccine strategy applied, these simulations show that exclusively implementing a control law that uniformly reduces travel 
based on $\bar{x}^k$, namely using \eqref{eq:control}-\eqref{eq:theta}, does
not have a significant impact on the total number of people infected.
Furthermore, in many cases increasing the strength  of the controller will cause the total number of recovered individuals to increase.
This phenomenon is illustrated in the plot at the bottom of Figure \ref{fig:control_no_vaccine} by the fact that as the x-axis increases (recall that the $\eta$ value corresponds to the strength of the controller), so does $\bar{r}^*$ 
(except for very small $\gamma$ values).
% \phil{[should we also remind the reader here that the controller is made stricter (or stronger) by increasing $\eta$?]}. 
% \phil{[Explain the intuition behind why this happens]}
This behavior is the consequence of at least two reasons: 1) the controller does not prevent infections from occurring inside the sub-populations (i.e. no lock downs), and 2) the controller in some cases `flattens the curve' which prolongs the outbreak and increases the integral under the curve (i.e. a higher number of total infections).
% While this control policy does not significantly impact the overall recovered population, it may be a good representation for some of the blanket travel restrictions put in place during the COVID-19 pandemic.

% One potential direct benefit from this controller stems from the flattening of the infection curve, prolonging the epidemic as seen in the top of Figure~\ref{fig:control_no_vaccine}. 
On the other hand, a reduction in the peak infected population 
% would  likely reduce 
reduces 
strain on healthcare systems, improving medical outcomes and decreasing fatality rates~\cite{kenyon2020flattening}, 
% Additionally, if there is a significant delay between the peak infected population for any two nodes, there is the potential to reallocate mobile healthcare such as military hospital ships and hospital tents. 
% Another note is that this policy increasingly extends the duration of the epidemic as the sensitivity to $\bar{x}^k$ increases, displayed in \edit{the top of} Figure~\ref{fig:control_no_vaccine}. 
% Through additional simulations, we observed that the controller is more effective at delaying the spread of the virus and prolonging the epidemic when the value of $p^x$ is lower.  In other words, reducing travel in response to the infection level is more effective at delaying the spread of the infection if infected individuals are less likely to travel.
% 
% While exclusively restricting travel does not significantly change the number of infected individuals during the epidemic, the behavior of the infected population seen in the top of Figure~\ref{fig:control_no_vaccine} 
and creates an opportunity for a vaccine to be more impactful.
In Figure~\ref{fig:control_vaccine} we simulate  the network with the same initial conditions and parameters as Figure~\ref{fig:control_no_vaccine}, adding the distribution of a vaccine, starting at time step $k=500$, which, for each sub-population, moves $0.1 \%$ percent of the susceptible proportion directly to the recovered proportion at each time step. In the bottom of Figure~\ref{fig:control_vaccine}, we show the mean infected proportion of the system for increasing levels of sensitivity as well as the total proportion of the recovered population who were infected by the disease~{($\delta \sum_{k\geq 0} \bar{x}^k$)}.
We see that a combined strategy of restricting flow and vaccine distribution can have a marked effect on the number of individuals infected while simultaneously reducing the peak infection level of the system.

Lastly, we implement the controller in \eqref{eq:control} with $\theta^k$ being strictly binary, that is, $\theta^k$ is either 0 or 1.
In Figure~\ref{fig:binary_controller_no_vaccine}, control strategies $\eta = 1,2$ shut down all travel after $50$ time steps, imitating a delay in decision making from policymakers while strategies $\eta = 3,4$ shut down travel immediately when the infection is first detected in any sub-population. 
A critical note is that even when travel is eventually completely closed, if it does not happen quickly enough then there will be almost no discernible impact on infection levels, as seen by comparing control strategies $\eta = 1,2$. 
These extreme cases illustrate that unless the infection is completely eradicated ($\bar{x}^k=0$) prior to reopening travel, it will always spread throughout the network after travel is reopened. When we tested the binary controller case $\eta = 3$ reopening travel after $\bar{x}^k<10^{-9}$, there was a second $\bar{x}^k$ peak around $k=7000$.
When a vaccine is distributed, this delayed second wave can be significantly mitigated as shown in Figure~\ref{fig:binary_controller_vaccine}. 
While costly, completely closing all travel between sub-populations can enable a vaccine to have a tremendous impact but only if
the initial response is not delayed.

\section{Conclusion} \label{sec:conclusion}

In this paper, we have constructed a networked discrete-time SEIR epidemic model that incorporates population flows, 
presented conditions under which the model is well-defined, and shown asymptotic convergence to the healthy states, the set of
equilibria. Additionally, we have proposed a control policy for restricting population flow  which can be interpreted as implementing travel restrictions/bans,
% based on the overall infection levels of the system, 
showed it is well defined, and illustrated its
behavior 
% s effectiveness 
via simulation. We have found 
that only restricting the flow of the population is typically insufficient to reduce 
% and 
% sometimes 
% can even increase, 
the total number of infections over the course of an epidemic. 
% \phil{[Explain why!!!]}.
More severe restrictions on the population flow can decrease the peak infection level, which can alleviate stress on healthcare facilities. Further, applying population flow restrictions
% , which can be interpreted as travel restrictions/bans, 
together with
a vaccination strategy can significantly reduce the total number of infections.

For future work we plan to incorporate the possibility of
infections occurring while individuals are traveling (i.e., infections occurring on the edges of the graph) as well as using real travel and infection data from the COVID-19 pandemic to learn the model parameters.
% It should be noted that 
Finally, note that
our 
% we use an SEIR epidemic 
model
% , which 
does not capture asymptomatic transmission of the virus, 
% which
% . However, the infectious behavior of COVID-19 
% may be better at capturing
a key component of 
the infectious behavior of COVID-19, therefore, developing a similar SAIR formulation 
% should be explored in 
remains open to future work.
% d with asymptomatic transmission included, such as an
% SAIR model, which is another area of possible
% future work.
% }

\normalem
\bibliographystyle{IEEEtran}
\bibliography{reference}


%\end{IEEEbiography}


%\newpage



%\vfill


\end{document}